\documentclass[onecolumn, referee]{svjour3}         

\usepackage{mathrsfs}

\usepackage[dvips,final]{graphicx}
\usepackage[dvips]{geometry} 
\usepackage{color} 
\usepackage{epsfig} 
\usepackage{latexsym}
\usepackage{pstricks}
\usepackage{ulem} 

\newtheorem{Def}{\em{Definition}}
\newtheorem{Theorem}{\em{Theorem}}

\usepackage[hyphens]{url}
\usepackage{algorithm}
\usepackage{algorithmic}
\usepackage{amssymb, amsmath}
\usepackage{mathrsfs}
\usepackage{natbib}

\newcommand{\Cmax}{C_{\max}}
\newcommand{\NP}{{\cal NP}}

\renewenvironment{proof}{\noindent {\bf Proof. }}{$\blacksquare$ \vskip 1em}

\def \Cmax {C_{\max}}
\def \Lmax {L_{\max}}
\def \NP {{$\cal NP$}}
\def \P {{$\cal P$}}

\begin{document}
\normalem 
\title{
How the structure of precedence constraints may change the complexity class of scheduling problems 
} 

\author{D. Prot \and O. Bellenguez-Morineau}
\institute{D. Prot \and O. Bellenguez-Morineau \at
Ecole des Mines de Nantes, 4 rue Alfred Kastler, B.P. 20722 F-44307 Nantes Cedex 3. France\\
\email{dprot@talend.com, odile.morineau@mines-nantes.fr}}

\date{Received: date / Accepted: date}

\maketitle

\begin{abstract} 

This survey aims at demonstrating that the structure of precedence constraints plays a tremendous role
on the complexity of scheduling problems. Indeed many problems can be \NP-hard when considering
general precedence constraints, while they become polynomially solvable for particular
precedence constraints. Add to this, the existence of many very exciting
challenges in this research area is underlined.

\end{abstract}


\section{Introduction}

Precedence constraints play an important role in many real-life scheduling problems.
For example, when considering the scheduler of a computer,
some operations have to be finished before some others begin. Other classical examples can be found
in the book of Pinedo (\cite{P12}). In the most general cases, precedence constraints can be represented by an
arbitrary directed acyclic graph. Nevertheless, in some cases, it is possible for precedence constraints to take
a particular form. For example, if a problem includes only precedence constraints related to assembling steps,
precedence constraints can be represented by a particular directed acyclic graph called intree.
The fact that the precedence graph takes a particular form may transform the complexity of the problem. Most of the time, adding precedence constraints makes problem harder, since the empty graph is included in many graph classes.
Yet, it may make the problem easier if it is not the case, for example for the class of graphs of bounded width.
For this reason, the idea of this survey is to consider the complexity results in scheduling theory according
to the structure of precedence constraints.

We assume that the reader is conversant with the theory of \NP-completeness, otherwise the book by \cite{GJ79} is a good entry point.
We will discuss, in the conclusion of this paper, the parameterized complexity. The reader can refer to the book of \cite{DF12} 
if needed.
\cite{LRK78} offer a large set of complexity results for scheduling problems with precedence constraints.
\cite{GLLKRK79} and \cite{LLRKS93} propose two surveys of complexity results for scheduling problems, but they are not
necessarily focused on precedence constraints.
For complexity results with the most classical precedence constraints (i.e., chains, trees, series-parallel and arbitrary precedence constraints),
the reader can refer to the book by ~\cite{B07}, and/or to the websites \cite{BK} and \cite{D16}.
Note that this website has been recently updated in order to include most of the results discussed in this article.
\cite{M89} proposes a very interesting survey dedicated to specific partial ordered sets, and studies their structure. 
Some applications to scheduling are also presented. Since this survey was conducted, many results arised in scheduling theory for specific precedence graphs,
we hence believe that a new survey would be beneficial for the scheduling community.
We restrict our field on purpose to complexity results and do not talk about approximations results (\cite{WS11}), despite the fact that many results
arised in this field recently in scheduling theory, such as \cite{S11} and \cite{LR16}. 
We believe that it is important to limit the scope of the survey,
in order to be as complete as possible in a given field.

The paper is organized as follows:
In Section~\ref{sec:order}, we introduce all the specific types of precedence constraints that will be studied in this paper, while scheduling notations are
recalled in Section~\ref{sec:notation}.
Section~\ref{sec:one} is dedicated to single-machine scheduling problems, and Sections~\ref{sec:paraNo} and~\ref{sec:paraYes}
are respectively dedicated to non-preemptive and preemptive parallel machines scheduling problems.
Each of these three sections is based on the following structure: we first give the polynomial results, then the
\NP-hard cases and last the most interesting open problems.
Finally in Section~\ref{sec:ccl} we give some concluding remarks.

\section{Special types of precedence constraints}\label{sec:order}

In this section we introduce, for the sake of completeness, the special types of precedence constraints that have already been studied in scheduling theory. 
Precedence constraints between jobs are easily modelled by a directed acyclic graph and we will use it as long as possible. Nevertheless,precedence relations may also be treated as partially ordered set (poset) in the terminology of order theory, and some definitions at the end of this section are much easier to understand from this point of view.

First, let us recall some basic graph theory definitions.
Let $G=(X,A)$ be a directed acyclic graph, where $X$ denotes the set of vertices and $A$, the set of arcs.

A directed acyclic graph is a collection of {\it chains} if each vertex has at most one successor and at most one predecessor.
An {\it inforest} (resp. {\it outforest}) is a directed acyclic graph where each vertex has at most one successor (resp. {\it predecessor}). 
An {\it intree} (resp. {\it outtree}) is a connected inforest (resp. outforest).
We call {\it forest} (resp. {\it tree}) a graph that is either an inforest or an outforest (resp. either an intree or an outtree).
An {\it opposing forest} is a collection of inforests and outforests.


For each vertex $x\in X$, we can compute its {\it level}  $h(x)$ which corresponds to the longest directed path starting from $x$ in $G$.
The height of a DAG, denoted $h(G)$, corresponds to the number of levels - 1 in this graph, as illustrated with Figure~\ref{fig:prec}.
Directed acyclic graphs of {\it bounded height} correspond to directed acyclic graphs where the height is bounded by a constant.

\begin{center}
\begin{figure}[htbp]
\begin{minipage}[c]{.6\linewidth}
\includegraphics[scale=0.2]{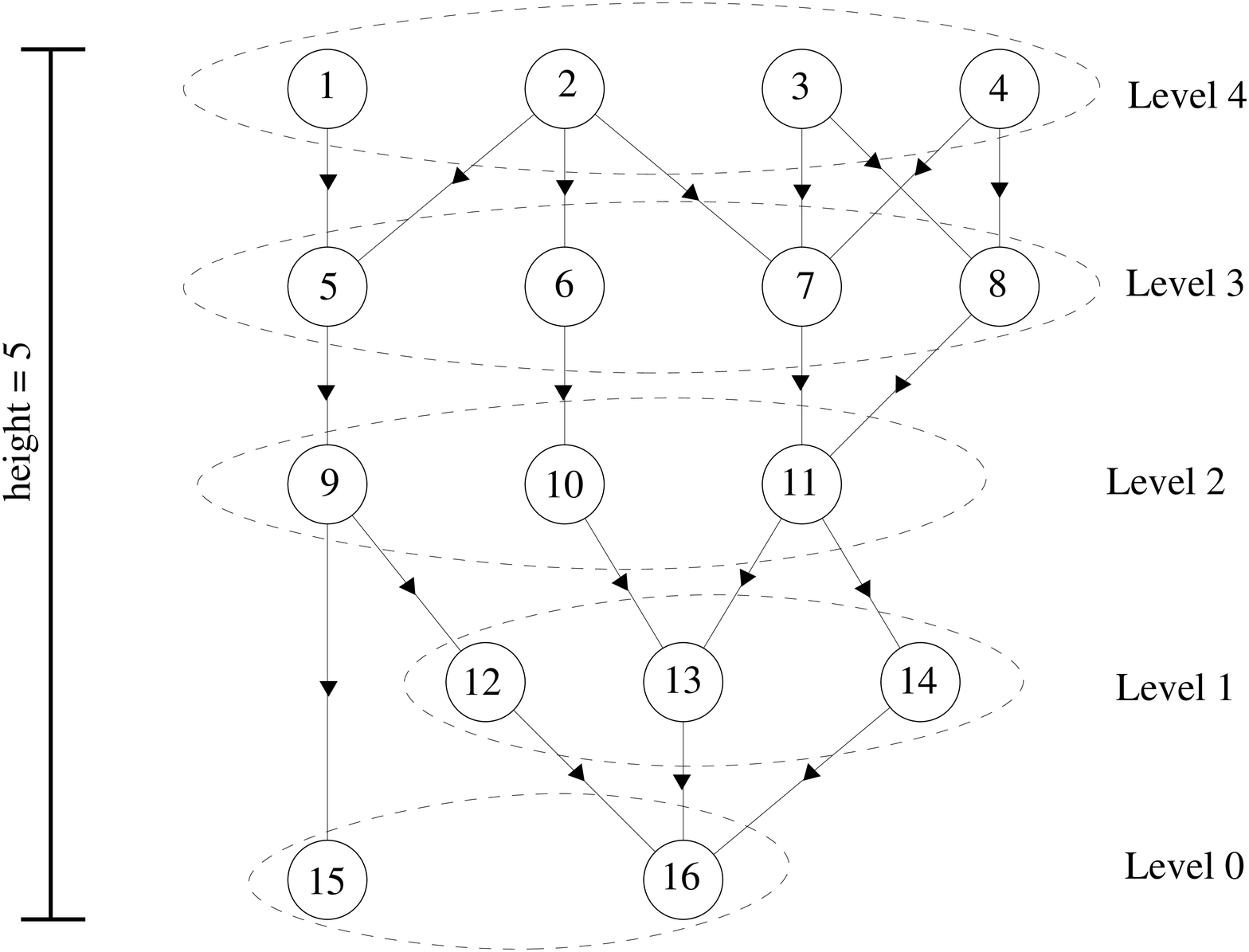}
\caption{The height of a DAG}
\label{fig:prec}
\end{minipage}\hfill
\begin{minipage}[c]{.35\linewidth}
\includegraphics[scale=0.2]{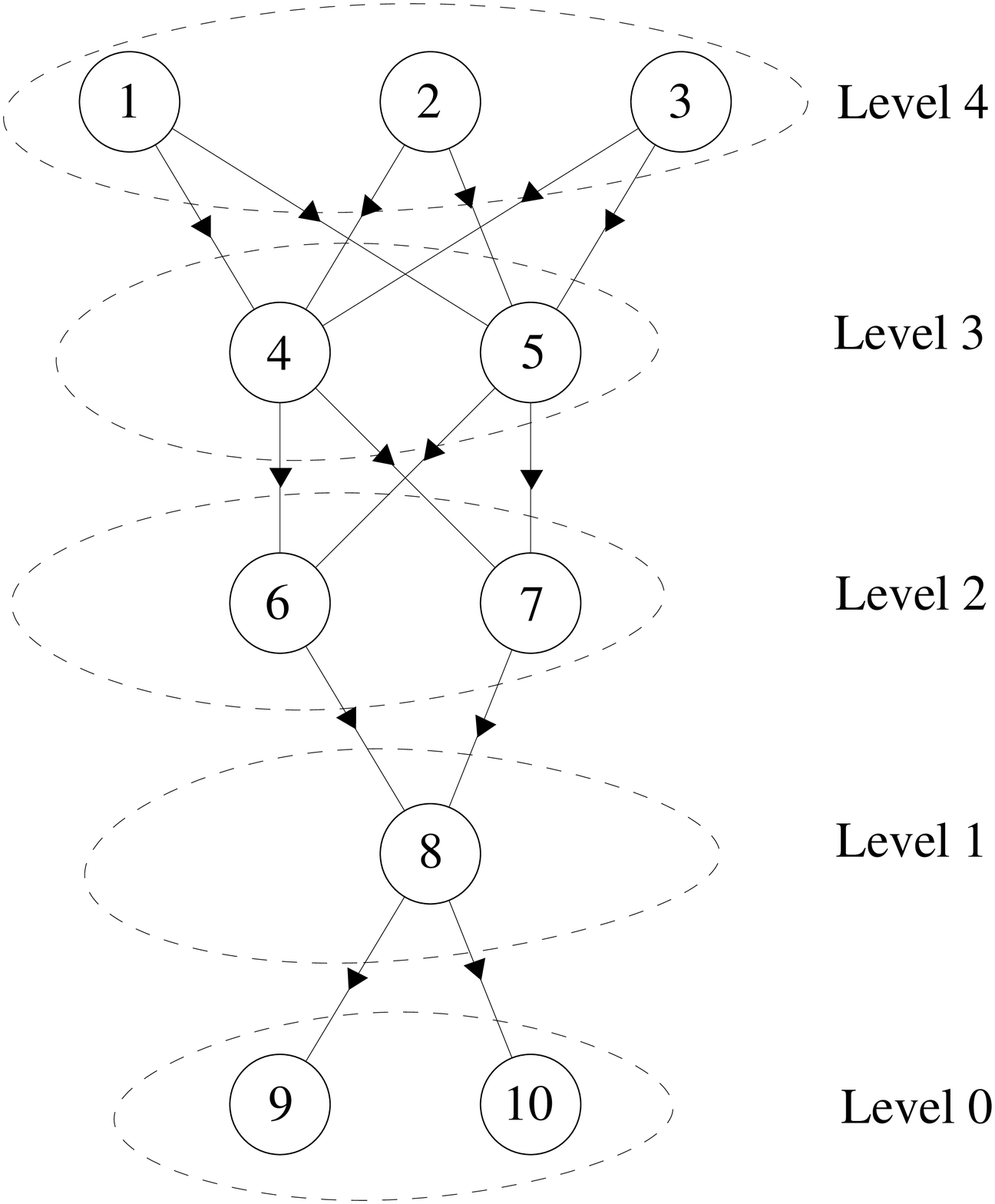}
\label{fig:levelOrder}
\caption{A level order graph}
\end{minipage}
\end{figure}
\end{center}

\begin{Def}\label{def:lo}
{\em (Level order graph)}
A directed acyclic graph is a level order graph if each vertex of a given level $l$ is a predecessor of all the vertices of level $l-1$ (see Figure~\ref{fig:levelOrder}).
\end{Def}

{Series-parallel} graphs (sp-graph) are defined in many ways, we use the inductive definition of~\cite{L78}.
\begin{Def}\label{def:spg}
{\em (sp-graph)}
A graph consisting of a single vertex is a sp-graph. Given two sp-graphs $G_1=(X_1,A_1)$ and $G_2=(X_2,A_2)$, 
the graph $G=(X_1\cup X_2, A_1 \cup A_2)$ is a sp-graph (this is called parallel composition).
Given two sp-graphs $G_1=(X_1,A_1)$ and $G_2=(X_2,A_2)$, 
the graph $G=(X_1\cup X_2, A_1 \cup A_2 \cup X_1 \times X_2)$ is a sp-graph (this is called series composition).
\end{Def}
An example is given in Figure~\ref{fig:spg}. Note that sp-graph can also be defined by the forbidden subgraph of Figure~\ref{fig:forbiddenSP}.

\begin{center}
\begin{figure}[htbp]
\begin{minipage}[c]{.55\linewidth}
\includegraphics[scale=0.2]{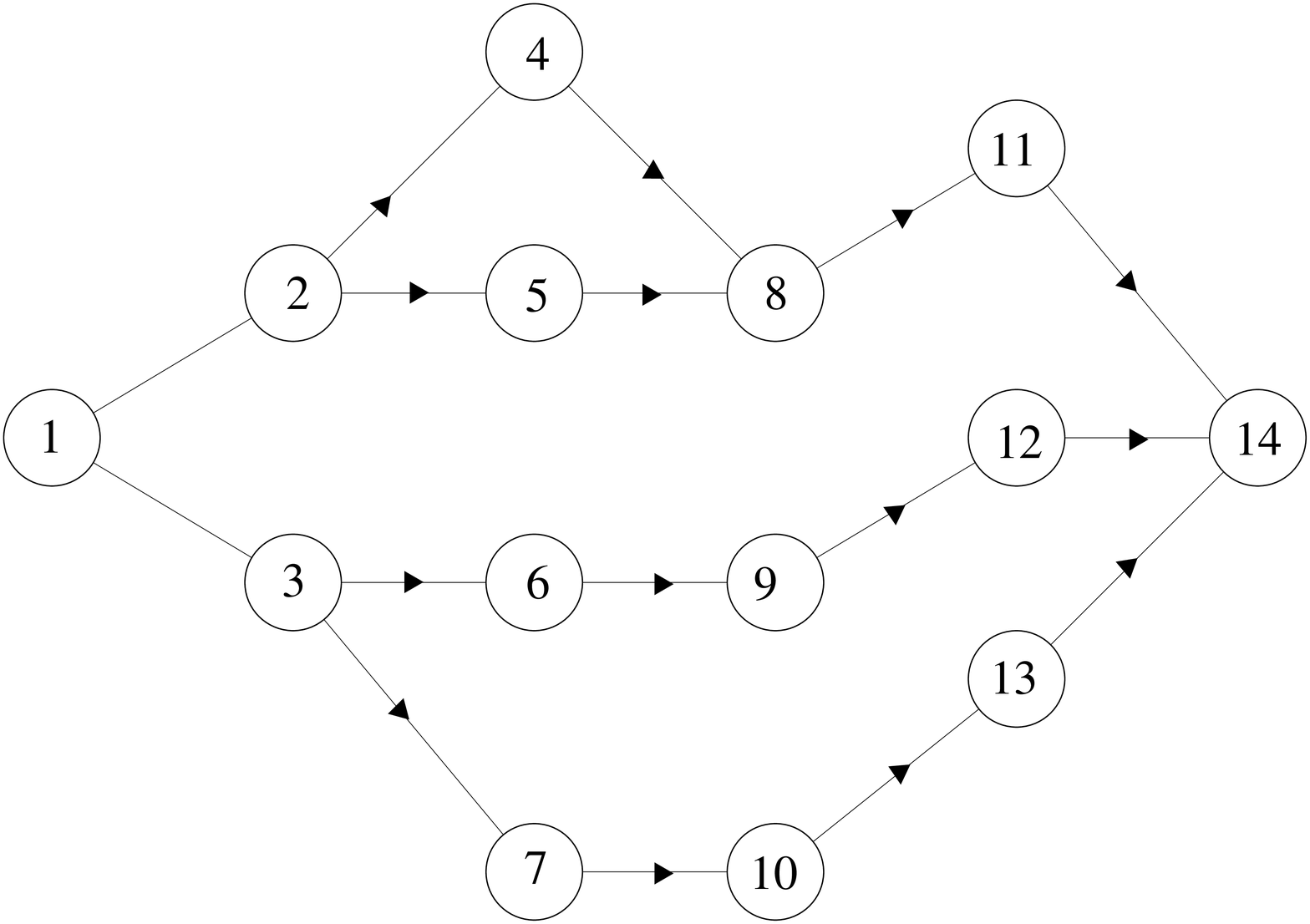}
\caption{A series-parallel graph}
\label{fig:spg}
\end{minipage}\hfill
\begin{minipage}[c]{.4\linewidth}
\includegraphics[scale=0.2]{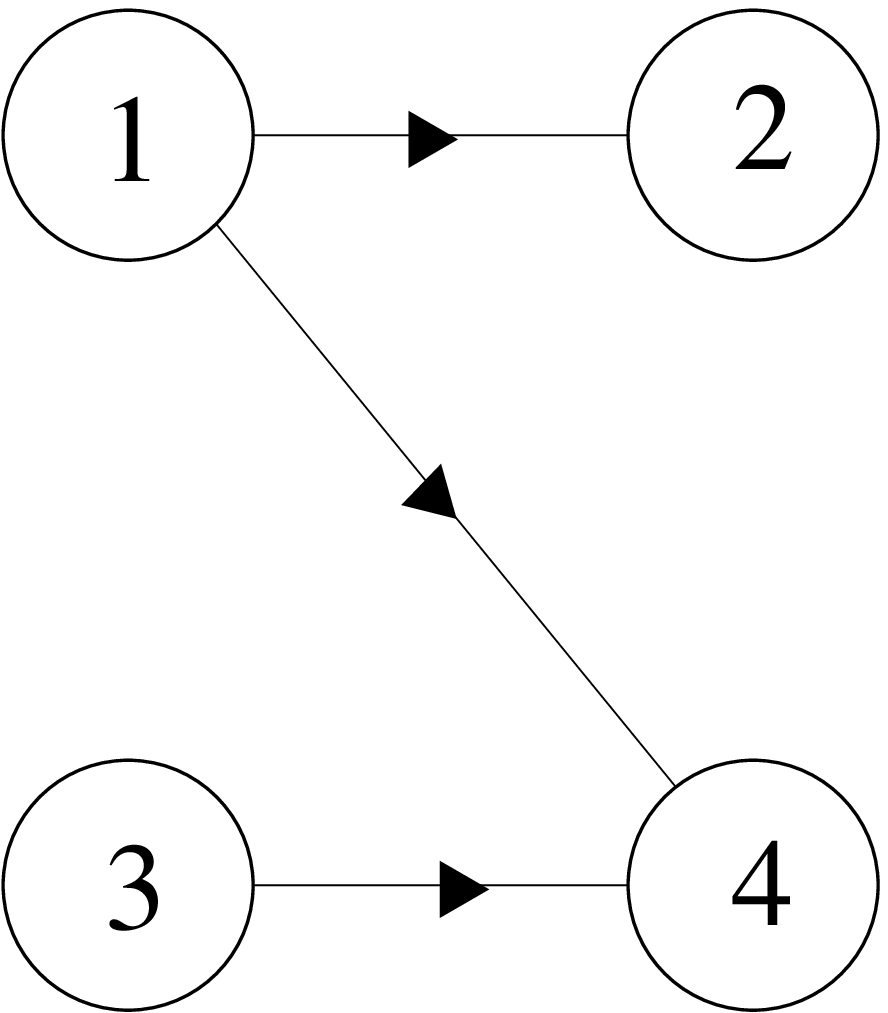}
\caption{The forbidden subgraph for sp-graphs}
\label{fig:forbiddenSP}
\end{minipage}
\end{figure}
\end{center}

A {divide-and-conquer-}graph (DC-graph) is a special sp-graph built using symmetries. It can be used to model divide-and-conquer algorithms (for example binary search, merge sort, number multiplication...) and is formally defined as follows:
\begin{Def}\label{def:dc}
{\em (DC-graph)}
A single vertex is a DC-graph; given two vertices $s$ and $t$ and $k$ DC-graphs $(X_1,A_1),\dots, (X_k,A_k)$, 
the graph $\left(\cup_{k=1}^nX_k \cup \{s\}\cup \{t\}, \cup_{k=1}^n \left(A_k \cup (s\times X_k) \cup (X_k \times t)\right)\right)$ is a DC-graph.
\end{Def}
This definition is illustrated with Figure~\ref{fig:DC}.

\begin{center}
\begin{figure}[htbp]
\begin{center}
\includegraphics[scale=0.2]{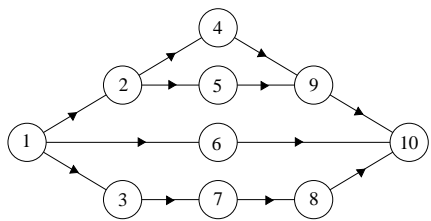}
\caption{A DC-graph}
\label{fig:DC}
\end{center}
\end{figure}
\end{center}

\begin{center}

\begin{figure}[htbp]
\begin{minipage}[c]{.7\linewidth}
\includegraphics[scale=0.2]{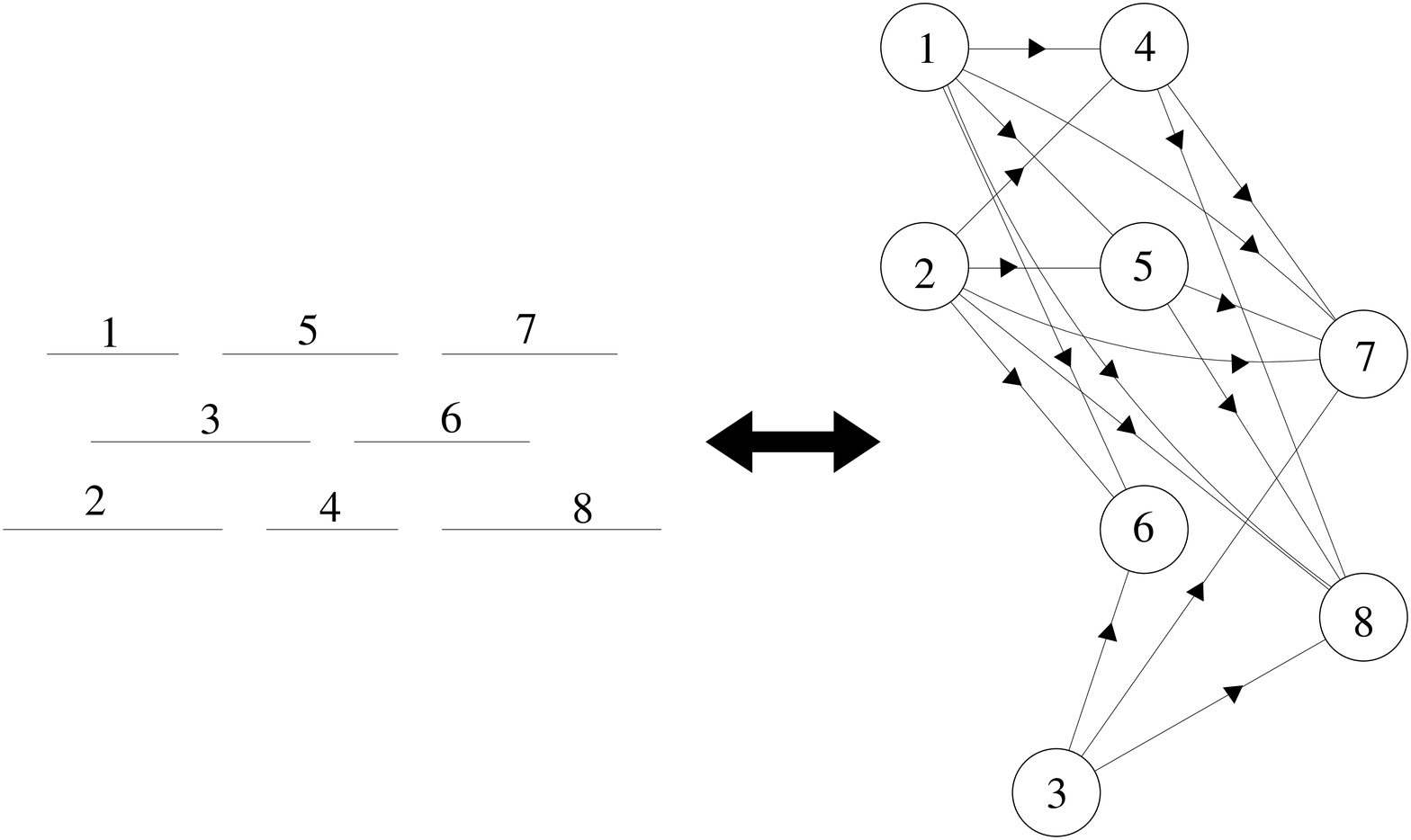}
\caption{A set of real intervals and the corresponding interval order graph}
\label{fig:interval}
\end{minipage}\hfill
\begin{minipage}[c]{.3\linewidth}
\includegraphics[scale=0.2]{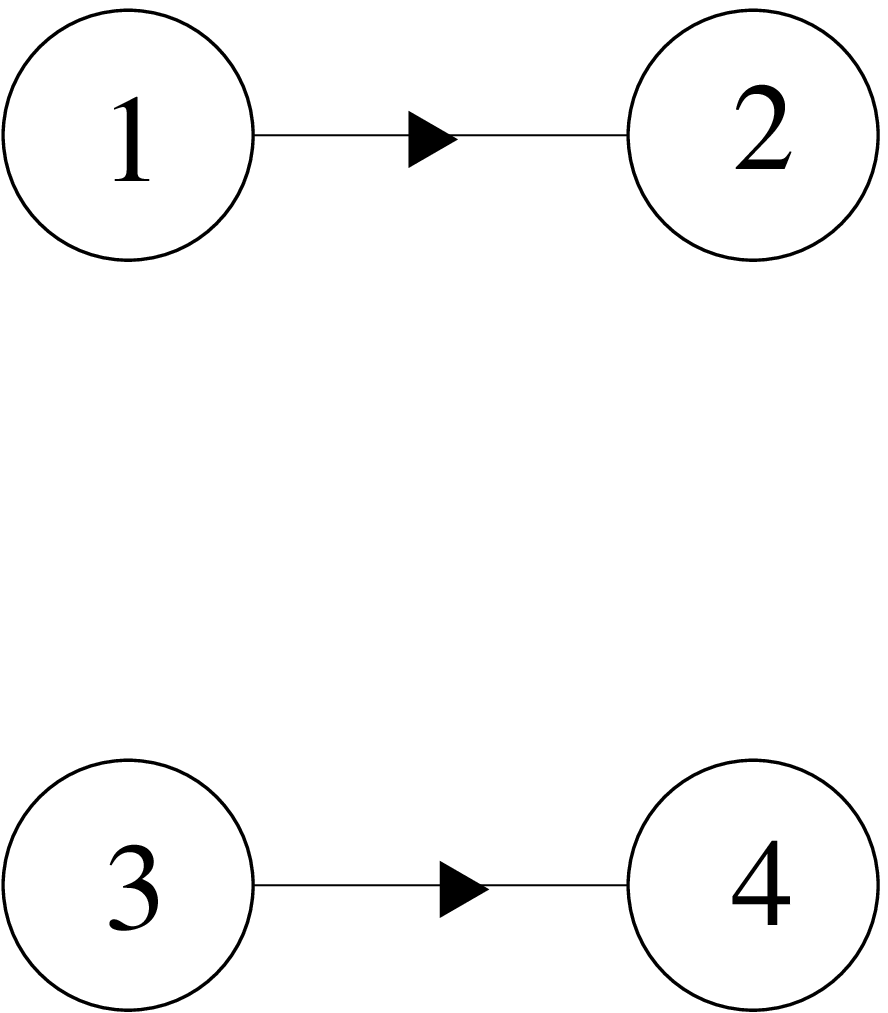}
\caption{The forbidden subgraph for interval order graphs}
\label{fig:forbiddenI}
\end{minipage}
\end{figure}
\end{center}

\begin{Def}\label{def:io}
{\em (Interval order graph)}
The interval order for a collection of intervals on the real line is the partial order corresponding to their left-to-right precedence relation, so that the interval order graph is the Hasse diagram of an interval order. 
such that for any two vertices $x$ and $y$, $(x,y)\in A \iff e_x \leq s_y$.
\end{Def}

An example is given in Figure~\ref{fig:interval}.
\cite{PY79} show that interval order graphs can also be defined by the forbidden induced subgraph presented in Figure~\ref{fig:forbiddenI}.
Larger classes of graphs were defined by forbidden subgraphs, such as {\em quasi-interval order graphs} and {\em over-interval order graphs}
(respectively in~\cite{M99} and \cite{CM05}), the quasi-interval order graphs being strictly included in over-interval order graphs. The corresponding
forbidden subgraphs are drawn in Figures~\ref{fig:forbiddenQIO} and~\ref{fig:forbiddenOI}.

\begin{center}
\begin{figure}[htbp]
\begin{minipage}[c]{.5\linewidth}
\includegraphics[scale=0.2]{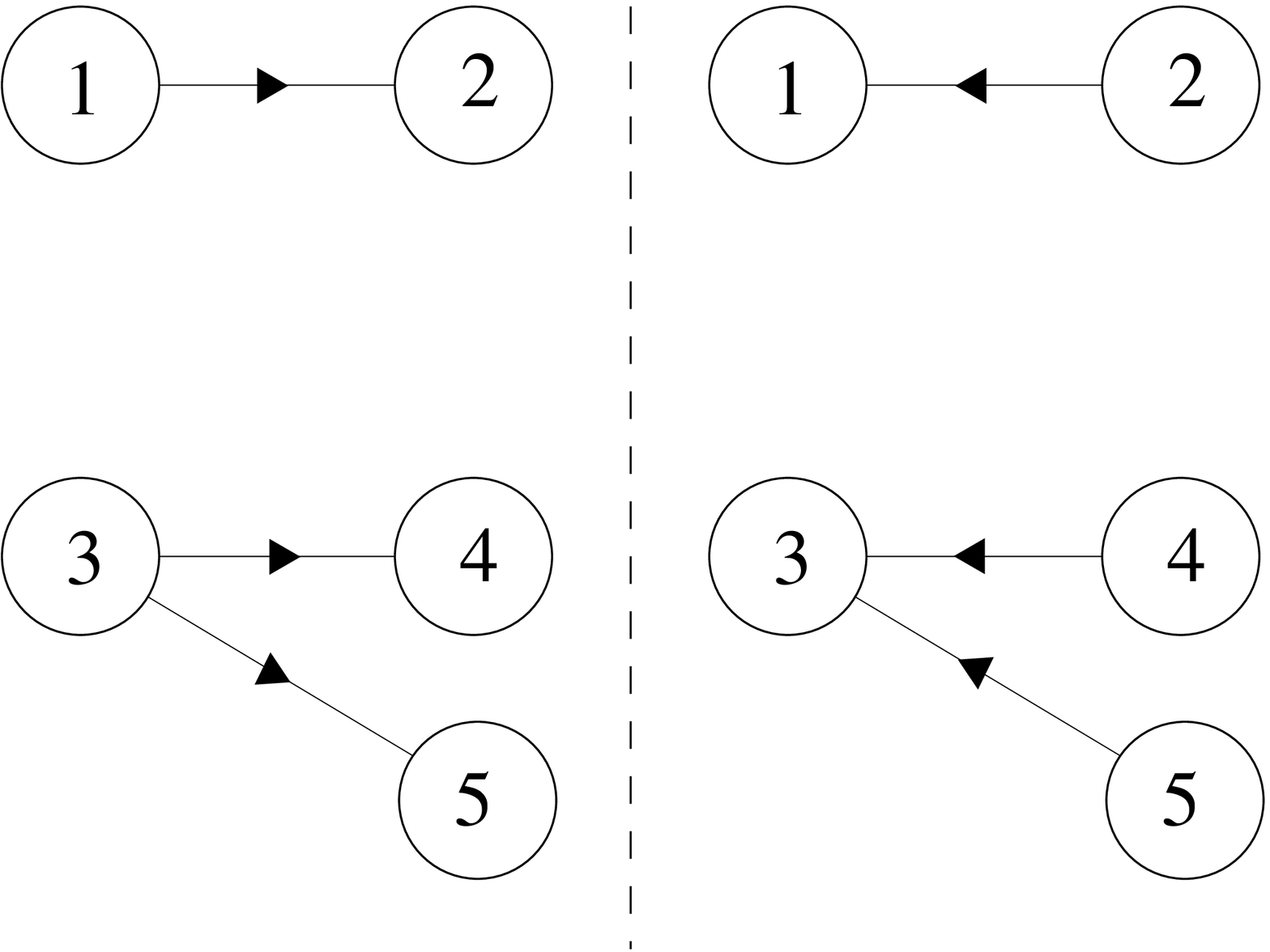}
\caption{The three forbidden subgraphs for quasi-interval order graphs}
\label{fig:forbiddenQIO}
\end{minipage}\hfill
\begin{minipage}[c]{.45\linewidth}
\includegraphics[scale=0.2]{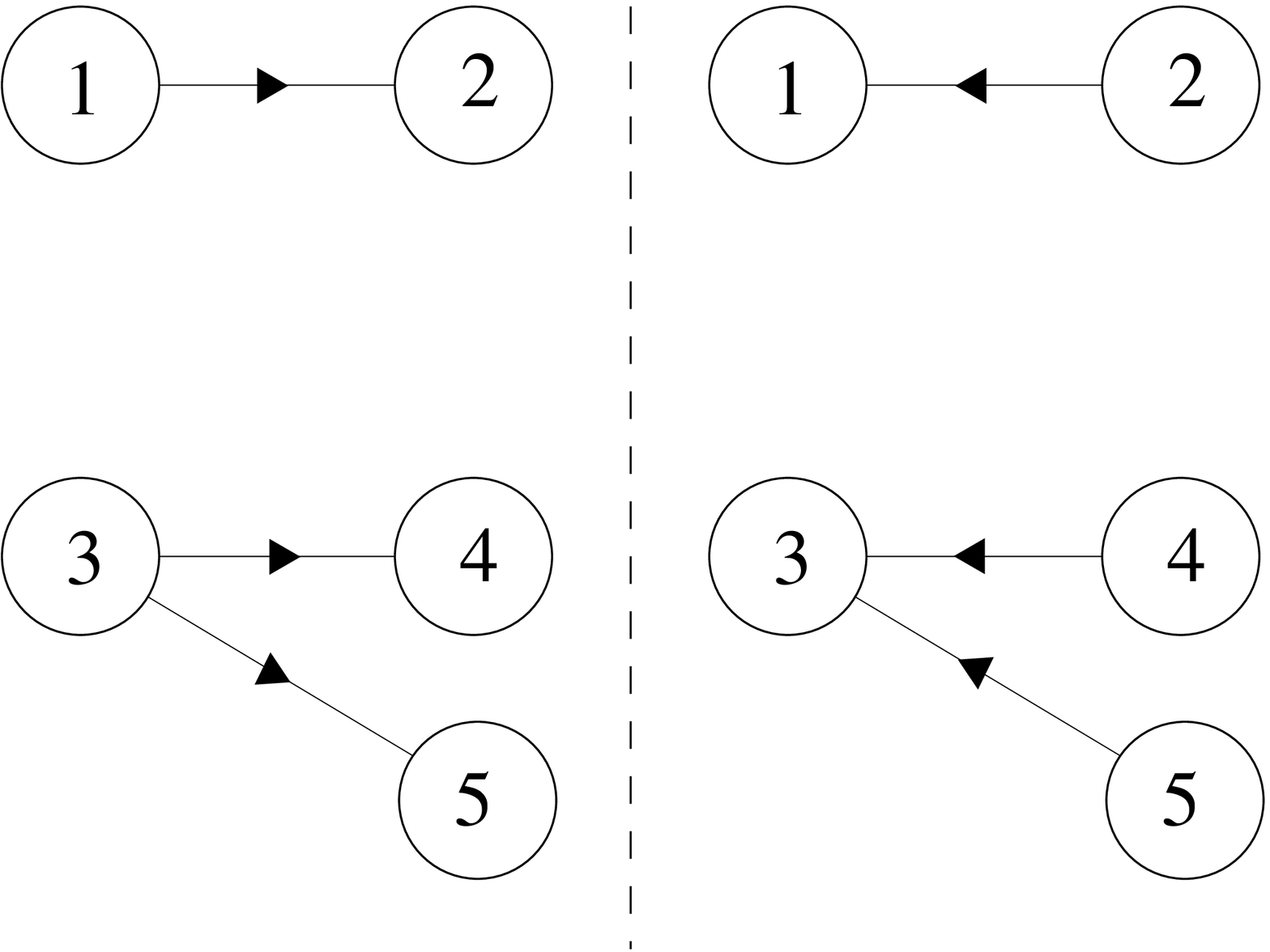}
\caption{The two forbidden subgraphs for over-interval order graphs}
\label{fig:forbiddenOI}
\end{minipage}
\end{figure}
\end{center}

To ease the reading, the following definitions will be given within the order theory paradigm. We will hence talk about a partial order set ${\cal P}=(X,\preceq_P)$ rather than a directed acyclic graph $G=(X,A)$ to describe the precedence graph, and a partial order $\preceq_P$ corresponds to the precedence constraints.

\begin{Def}\label{def:antichain}
{\em (Antichain)}
Given a partial order set ${\cal P}=(X,\preceq_P)$, an antichain is a subset $S$ of $X$ such that any two elements of $S$ are incomparable.
\end{Def}

\begin{Def}\label{def:width}
{\em (Width)}
Given a poset ${\cal P}=(X,\preceq_P)$, the width of a poset is the size of a maximum antichain.
\end{Def}

By extension, for a given directed acyclic graph $G=(X,A)$ we define the width of the graph to be the width of the corresponding poset, denoted by $w(G)$.

The ${\cal A}_m-$order (first introduced in~\cite{MQ97}) contains the over-interval order for any integer $m\geq 2$ and is defined in the following way:

\begin{Def}\label{def:amo}
{\em (${\cal A}_m-$order)}
Let ${\cal P}=(X,\preceq_P)$ be a poset. For any two antichains $A$ and $B$ of size at most $m$, let us define the four
sets :
$\max (A,B)=\{x \in A \cup B | \exists y \in A\cup B, y \preceq_P x\}$,
$\min (A,B)=\{x \in A \cup B | \exists y \in A\cup V, x \preceq_P y\}$,
$\overline{\max}(A,B) = (A \cap B) \cup \max (A,B)$, and
$\overline{\min}(A,B) = (A \cap B) \cup \min (A,B)$.
$\preceq_P$ is an ${\cal A}_m$ order if and only if there do not exist two antichains $A$ and $B$ of size $m$ at most, such that $|\overline{\max}(A,B)| \geq m+1$
or $|\overline{\min}(A,B)| \geq m+1$.
\end{Def}

We call ${\cal A}_m-$order graph a directed acyclic graph for which the set of arcs corresponds to an ${\cal A}_m-$order.

\begin{Def}\label{def:linExt}
{\em (Linear extension)}
Given a partial order $\preceq_P$ over a set $X$, a linear extension of $\preceq_P$ over $X$ is a total order respecting $\preceq_P$.
\end{Def}

\begin{Def}\label{def:kdim2}
{\em (Dimension)}
The dimension of a poset ${\cal P}=(X,\preceq_P)$ is the mimimum number $t$ of
linear extensions $\preceq_1,\dots, \preceq_t$ such that $ x \preceq_P y \iff \forall l \in 1..t, x \preceq_l y$. In other words, 
if $x||y$ ($x$ and $y$ are incomparable in $\preceq_P$), then there are
at least two linear extensions, one with $x \preceq y$ and another one with $y\preceq x$.
\end{Def}

An interesting point is that series-parallel graphs are strictly included in directed acyclic graphs of dimension 2.

The {fractional dimension} of a poset is extending the notion of dimension (see~\cite{BS92}).
\begin{Def}\label{def:fdim}
{\em (Fractional dimension)}
For any integer $k$, $t(k)$ denotes the minimum number of linear extensions such that 
for any two incomparable elements $x$ and $y$, there are at least $k$ extensions with $x\preceq y$ and 
$k$ with $y\preceq x$. The {fractional dimension} is the limit  $\{t(k)/k\}$ as $k$ tends to $+\infty$.
\end{Def}

The diagram in Figure~\ref{fig:inclusion} provides a better overview of the existing inclusion between the different classes.

\begin{center}
\begin{figure}[htbp]
\includegraphics[scale=0.25]{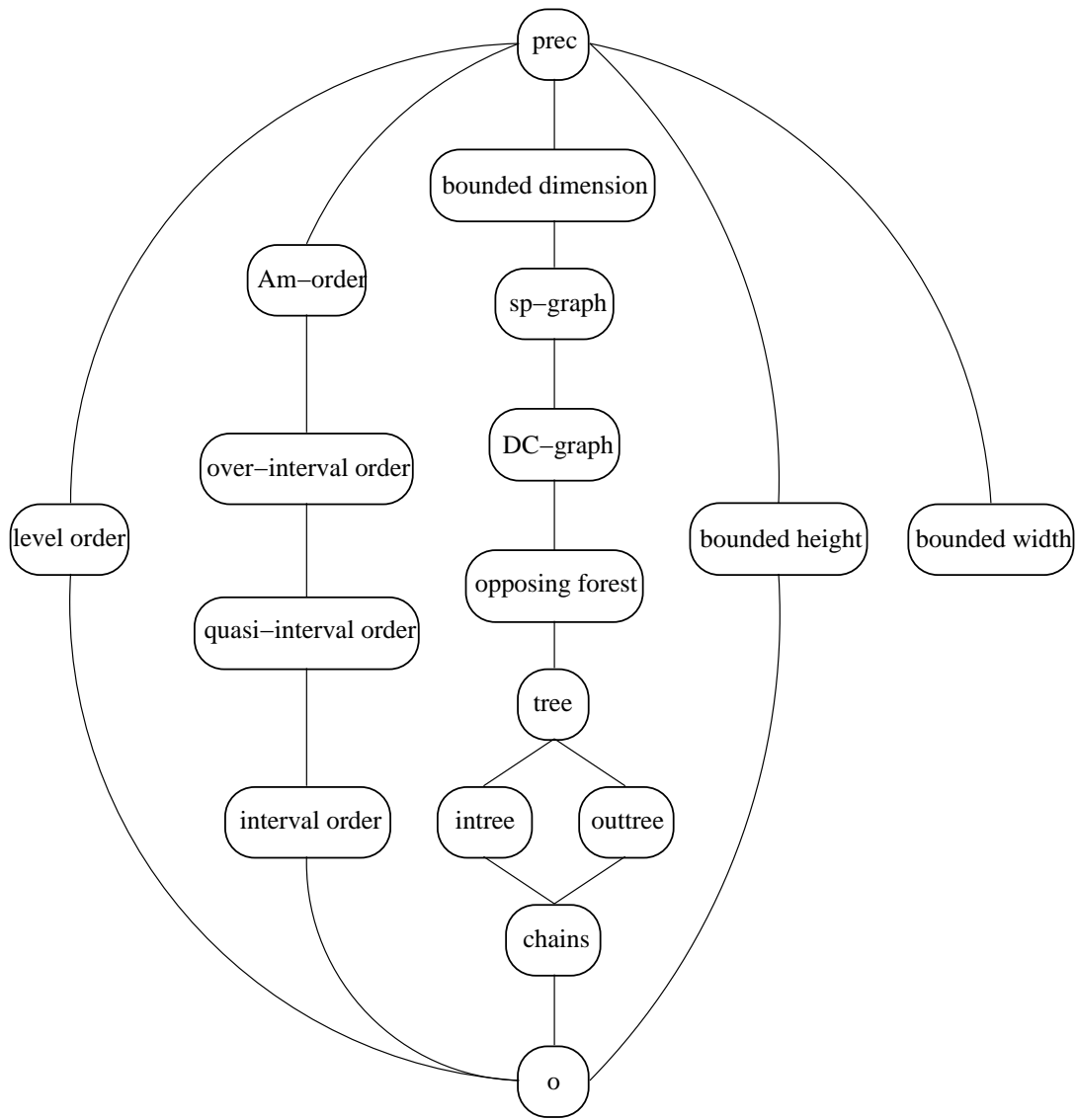}
\caption{Hasse diagram for different classes of directed acyclic graphs}
\label{fig:inclusion}
\end{figure}
\end{center}

\section{Scheduling notation}\label{sec:notation}

In this paper, we will use the standard $\alpha|\beta|\gamma$ notation introduced in~\cite{GLLKRK79}, and updated in~\cite{B07}. We define below the different notations used all along the paper.

The $\alpha$-field is used for the machine environment. $\alpha=1$ corresponds to a single machine problem;
if $\alpha=Pm$, there is a fixed number $m$ of identical parallel machines. 
If this number is arbitrary, it is noted $\alpha=P$. 
Similarly, if $\alpha=Qm$ (resp. $\alpha=Q$), it corresponds to a fixed (resp. arbitrary) number
of uniform parallel machines, i.e., each machine $i$ has a speed $s_i$ and the processing time of a job $j$
on machine $i$ is equal to $p_j/s_i$.

The field $\beta \subset \{ \beta_1, \beta_2, \beta_3, \beta_4\}$ describes the jobs characteristics, the possible entries that we
will deal with are the following ones:

\begin{itemize}
\item
$\beta_1 \in \{ pmtn, \circ\}$ : $pmtn$ means that preemption is allowed, i.e., a job may
be interrupted and finished later. If $\beta_1=\circ$, preemption is forbidden.
\item
$\beta_2$ describes the precedence constraints. If $\beta_2=\circ$, there is no precedence constraint, whereas $\beta_2=prec$ means that the precedence graph
is a general directed acyclic graph. This field can take many values according
to the structure of the directed acyclic graph, as presented in the previous section. For sake of completeness,
all the acronyms are recalled hereinafter: $chains$, $intree$, $outtree$, $opp. forest$ (opposing forest), $io$ (interval order graph),
$qio$ (quasi-interval order graph), $oio$ (over-interval order graph), $sp-graph$ (series-parallel graph), $DC-graph$ (divide-and-conquer graph), $lo$ (level order graph), $h(G)\leq k$ (directed acyclic graph with height bounded by $k$), $dim \leq k$ (directed acyclic graph with dimension bounded by $k$), $fdim \leq k$ (directed acyclic graph with fractional dimension bounded by $k$), $w(G) \leq k$ (directed acyclic graph
with width bounded by $k$).

It appears there that classical literature often abuses of terms $intree$ and $outtree$ to handle in fact $inforest$ and $outforest$, that can be mixed in  $opp. forest$ for some problems. 
\item 
$\beta_3 \in \{r_j, \circ\}$ : if $\beta_3=r_j$, each job $j$ has a given release date. If $\beta_3=\circ$, the release date is $0$ for each job.
\item
$\beta_4$ represents the processing time of a job $j$.
If $\beta_4=\circ$, there is one processing time $p_j$ for each job $j$. 
If $p_j=p$, all the jobs have the same processing time $p$.
When $p_j=1$, we use the acronym UET which stands for Unit Execution Time.
In some cases, the processing time may increase or decrease with either the position of the job or the starting time of the job. If the job is in position $r$ on a machine, the processing time will be noted $p_j^{[r]}$. If the processing time is depending of the starting time $t$ of the job, it will be written $p_j^{(t)}$. 
\end{itemize}

The $\gamma-$field is related to the objective function of the problem. Let $C_j$ be the completion time of job $j$. The {\it makespan} is defined by $\Cmax=\max_j C_j$
and the {\it total completion time} by $\sum_j C_j$. 
It is clear that, in general, makespan and total completion time are not equivalent. Nevertheless, for some problems, we can show that there exists an {\it ideal} schedule,
which both minimizes makespan $C_{max}$ and total completion time $\sum C_j$.
Due-date related objectives are also studied; if $d_j$ is the due date of job $j$, then the {\it lateness} of job $j$
is defined by $L_j=C_j-d_j$. The {\it tardiness} is $T_j=\max(0,C_j-d_j)$ and the {\it unit penalty} $U_j$ is equal to one if $C_j > d_j$ and to zero otherwise.
We then can define the {\it maximum lateness} $L_{max} = \max L_j$, the {\it total tardiness} $\sum T_j$, the {\it total number of late jobs} $\sum U_j$.
A weight $w_j$ may also exist for each job $j$, leading to the corresponding objective functions: the {\it total weighted completion time} $\sum w_jC_j$, the {\it total
weighted tardiness} $\sum w_j T_j$ and the {\it total weighted number of late jobs} $\sum w_j U_j$.
All the functions presented so far are {\it regular} functions, i.e., they are non-decreasing with $C_j$.

\section{Single machine problems}\label{sec:one}

For single machine problems, most of the interesting results for this survey are related to the total weighted completion time $\sum w_j C_j$.
It is mainly due to the fact that $1|chains,p_j=1|\sum U_j$ (see~\cite{LRK80}) and $1|chains,p_j=1|\sum T_j$ (see~\cite{LY90}) are already \NP-hard,
while $1|prec,r_j|\Cmax$ is solvable in polynomial time.
Nevertheless, as we will see in Section \ref{sec:oneOpen}, there are some interesting open problems for other criteria when considering preemption.


\subsection{Polynomial cases}

\cite{L78} uses Sidney's theory (\cite{S75}) to derive
a polynomial-time algorithm to solve problem $1|sp-graph|\sum w_iC_i$. 
The most important results are based on the concept of a module
in a precedence graph $G=(N,A)$:
A non-empty subset $M\subset N$ is a module if, for each job $j\in N-M$,
exactly one of the three conditions holds:
1. $j$ must precede every job in $M$,
2. $j$ must follow every job in $M$,
3. $j$ is not constrained to any job in $M$.
Using this concept leads to a very powerful theorem stating that there exists
an optimal sequence consistent with any optimal sequence of any module.

A large improvement on this problem has been done recently, by using
order theory and by proving that the problem is a special case of the
vertex cover problem.
More precisely, in~\cite{CS05} the authors conjecture
that $1|prec|\sum w_j C_j$ is a special case of vertex cover and
prove that, under this conjecture, the problem $1|prec|\sum w_j C_j$
is polynomial if the precedence graph is of dimension 2.
In~\cite{AM09}, the authors prove the conjecture and hence the result provided
in~\cite{L78} is considerably extended (since series-parallel graphs are
strictly included in DAGs of dimension 2).

Using the same methodology as in~\cite{L78}, \cite{WNC08} extend these results to the case where
jobs are deteriorating, i.e., processing times are an increasing function of their starting time, and 
they show that $1|sp-graph,p_j^{(t)}=p_j(1+at)|\sum w_jC_j$ and $1|sp-graph,p_j^{(t)}=p_j+\alpha_jt|\Cmax$ can be solved in a polynomial time, 
where $p_j$, $a$, and $\alpha_j$ are positive constants and $t$ is the starting time of the job.  
The same framework is used in~\cite{WW13} for position-dependent processing times,
and the authors show that $1|sp-graph,p_j^{[r]}=a_j+b_j r|\Cmax$ and $1|sp-graph,p_j^{[r]}=a_j-b_j r|\Cmax$, where $r$ is the position of the job, $a_j$ and $b_j$ positive constants,
are polynomially solvable.
\cite{GPSW08} propose a more general framework than the one introduced in~\cite{WNC08}, that was first presented in~\cite{MS79}. They extend it to different
models with deterioration and learning, that were aimed at minimizing either the total weighted completion time or the makespan, with series-parallel
precedence constraints.

\subsection{Minimum {\cal NP}-hard cases}

The problem $1|prec|\sum w_jC_j$ is known to be \NP-hard (see~\cite{LRK78}). 
In order to identify for which precedence graph the problem may be polynomially solvable, 
we present here the minimal (with respect to precedence constraints) \NP-hard
cases: the problem is still ${\cal NP}-$hard even if
the precedence graph is:

\begin{itemize}
\item of indegree at most 2 (see~\cite{LRK78}). Note that this result also stands with equal weights (i.e., $w_j=1$)
\item of bounded height (proof is straightforward by using the transformation
proposed in~\cite{L78}).
\item an interval order graph (see~\cite{AMMS11}). This is a strong difference with parallel machine results for
which interval order graphs provide often polynomial algorithms, as we will see in the dedicated section.
\item of fractional dimension (see Definition~\ref{def:fdim}) greater or equal to $3$ (see~\cite{AMMS11}). 
\end{itemize}

\subsection{Open problems}\label{sec:oneOpen}

The problem $1|prec|\sum w_iC_i$ has been widely studied, and the boundary between polynomial and \NP-hard cases is globally well-defined.
However, there are still boundaries to defined, as illustrated by the following examples. 
First,  if the fractional dimension $fdim$ of the precedence graph lies in interval $]2,3[$, the problem is open (it is polynomially solvable if $fdim \leq 2$
since the fractional dimension of a DAG is less than or equal to the dimension of this DAG and the problem is solvable in polynomial time if the precedence graph is of dimension 2).
This is an interesting question but the gap is rather limited. A wider open question 
is when we consider the problem with equal weights, i.e., $1|prec|\sum C_i$. The problem is still \NP-hard,
even if the precedence graph is of indegree at most 2 (see~\cite{LRK78}). Nevertheless, we may hope that the problem becomes polynomial
for precedence graphs with dimension larger than two. It is also possible that the problem is solvable in polynomial time if the precedence graph is an interval order graph, and even for larger classes like quasi-interval order graphs and over-interval order graphs.

For the criteria related to tardiness $T_j$ and unit penalty $U_j$, it is surprising to see that preemptive problems with precedence constraints have not been studied yet. More precisely, \cite{TNC06} have shown that $1|pmtn, r_j, p_j=p|\sum T_j$ is solvable in polynomial time, but it is still open whether adding precedence constraints leads this problem to be \NP-hard or not; the problem remains also open when considering the more general criterion $\sum w_j T_j$.
The same outline appears when considering unit penalty:  $1|r_j;pmtn;p_j=p|\sum w_jU_j$ is solvable in polynomial time (with an algorithm in $O(n^{10})$ by \cite{B99},
and in $O(n^4)$ by \cite{BCDJV04}); nevertheless, nothing has been shown when adding precedence constraints to the problem. 
The first research avenue is to study this problem including Chains as a first step and determine whether the problem is polynomial or not.

\section{Parallel machines without preemption}\label{sec:paraNo}


When considering non-preemptive scheduling problems, whatever the structure of precedence graphs, we will mainly focus on 
problems with equal processing times, since problems $P2||C_{max}$ and $P2|chains|\sum C_j$ are already NP-hard (see~\cite{LRKB77} and \cite{DLY91}).

\subsection{Polynomial cases}

\subsubsection*{Makespan criterion and arbitrary number of machines}

Three seminal works on parallel machines with precedence constraints are the approaches of~\cite{H61}, ~\cite{PY79} and \cite{M89}
in which the authors are respectively dealing with trees, interval order graphs and graphs of bounded width.

In~\cite{H61}, the author proves that problem $P|tree,p_j=p|\Cmax$ is polynomially solvable by a list scheduling algorithm where the
highest priority is given to the job with the highest level (this strategy is called HLF, for Highest Level First).
It is unlikely to find a precedence graph that strictly includes trees for which the problem is solvable in polynomial time, since it was proven
in~\cite{GJTY83} that scheduling opposing forest is \NP-hard. The hardness of the latter problem is mainly due
to the arbitrary number of machines, since it is solvable in polynomial time for any fixed number of machines.

In~\cite{PY79}, the authors prove that problem $P|io, p_j=p|\Cmax$ is polynomially solvable by a list scheduling algorithm where the
highest priority is given to the job with the largest number of successors.
This result has been improved twice; first \cite{M99} shows that the same algorithm gives an optimal solution if interval order graphs
are replaced by quasi-interval order graphs, that properly contains the former. \cite{CM05} show that the same result does not stand
for over-interval order graphs, but the Coffman-Graham algorithm (see~\cite{CG72}) can be applied to solve problem $P|oio,p_j=p|\Cmax$ to optimality.

In~\cite{M89}, the author studies the problem with bounded width, equal processing times and the makespan criterion. He shows that the problem
 can be solved in polynomial time by using dynamic programming on the digraph of order ideals. \cite{MT99} proposed their own approach to solve the problem $P|w(G)\leq k, r_j, p_j=1|f$ for any regular function $f$. More precisely, their algorithm consists in searching a shortest path in the related transversal graph.

When adding release dates, the problem is already \NP-hard for intrees, yet it is solvable in polynomial time for outtrees (see~\cite{BGJ77}, and \cite{M82} for a linear algorithm for the latter problem).
Note that there is a strong relationship between scheduling with release dates and $\Cmax$ criterion and scheduling with $L_{max}$, 
by simply looking at the schedule in the reverse way, and reversing the precedence constraints. Hence problems
$P|intree,p_j=p|L_{max}$ is polynomially solvable and $P|outtree,p_j=p|L_{max}$ is \NP-hard.

\cite{KRP09} recently opened new perspectives, since they show that $P|DC-graph,p_j=p|\Cmax$ is solvable in polynomial time. 
They more precisely proved that the Highest Level First strategy (used in~\cite{H61}) solves the problem to optimality when
the precedence graph is a divide-and-conquer graph (see Definition~\ref{def:dc}).

\subsubsection*{Makespan criterion and fixed number of machines}

Let us now focus on the case where the number of machines is fixed. 
Recall that the problem $Pm|prec,p_j=p|\Cmax$ is still open (this problem is known as [OPEN8] in the book by~\cite{GJ79}) for $m\geq 3$, and it was solved for $m=2$ in~\cite{CG72}.
For opposing forests, \cite{GJTY83} propose an optimal polynomial algorithm (of complexity $O(n^{m^2+2m-5}\log n)$) that consists in a divide and conquer approach, that uses the HLF algorithm as a subroutine. 
A new algorithm with complexity $O(n^{2m-2} \log n)$ has been proposed by~\cite{DW85}, who also show that
the problem is polynomially solvable for level order graphs (that are strictly included in series-parallel graphs). \cite{DW84} solve the case where the precedence graph is of bounded height by proposing an algorithm of time complexity $O(n^{h(m-1)+1})$.
Recently, \cite{AM06} show that $Pm|prec,p_j=p|\Cmax$ is solvable in polynomial time when the precedence graph is of bounded height and the maximum degree is bounded. This result is in fact a special case of the one proposed in~\cite{DW84}.

\subsubsection*{Other criteria and/or machine environment}

The other well-studied criterion for parallel machine environment is the total completion time $\sum C_j$.
One of the reasons for that is that, for some problems, it is equivalent to solve the total completion time and the makespan 
since they admit an ideal schedule.
For example, ideal schedules exist when considering two machines, arbitrary precedence constraints
and equal processing times, hence problem $P2|prec,p_j=p|\sum C_j$ is polynomially solvable with CG-algorithm (see~\cite{CG72}). When adding release dates, \cite{BT01} show that $P2|prec,r_j,p_j=1|\sum C_j$ is solvable in polynomial time by reducing it to a shortest path problem.  

For an arbitrary number of machines, if the precedence graph forms an outtree and the processing times are UET, 
the same result holds and hence $P|outtree,p_j=p|\sum C_j$ is solvable in polynomial time.
Note that problem $P3|intree,p_j=p|\sum C_j$ is not ideal, see~\cite{HL06} for a counterexample.
Nevertheless, for any fixed number of machines, problem $Pm|intree,p_j=p|\sum C_j$ is solvable in polynomial time (see~\cite{BBKT04}).
Adding release dates maintains the same property: the algorithm proposed in~\cite{BHK02} solves simultaneously
problems $P|outtree,r_j,p_j=1|\Cmax$ and $P|outtree,r_j,p_j=1|\sum C_j$. An improvement of this algorithm has
been proposed in~\cite{HL05}. 

For interval order graphs, \cite{M89} notices that $P|io,p_j=p|\sum C_j$ is solvable in polynomial time since the proof of the algorithm for $P|io,p_j=p|\Cmax$ (in \cite{PY79}) only
uses swaps between tasks, and this property is verified by the makespan and the total completion time.
It has been noticed recently that the same result holds for overinterval orders (that properly contain interval orders),
the problem admits also an ideal solution, so $P|oio,p_j=p|\sum C_j$ is solvable in polynomial time (see \cite{W15}).

When considering uniform parallel machines, only few results are available;  
problem $Q2|chains,p_j=p|\Cmax$ is solvable in polynomial time (see~\cite{BHK99}). 
If one processor is going $b$ times faster than the other (with $b$ an integer), the problem
$Q2|tree,p_j=p|\Cmax$ is also polynomially solvable (see~\cite{K89}); the problem is also
ideal, and hence  $Q2|tree,p_j=p|\sum C_j$ is also solvable in polynomial time.

\subsection{Minimum {\cal NP}-hard cases}

The interesting results for this survey are the \NP-hardness of $P|prec,p_j=p|\Cmax$ and $P|prec,p_j=p|\sum C_j$ (see~\cite{LRK78}).
Note that the proof for the two problems also holds when the precedence graph is of bounded height, and that the problem $P|opp. forest,p_j=p|\Cmax$ is
also \NP-hard (see~\cite{GJTY83}).
When adding release dates, the corresponding problem is already \NP-hard for intrees, for both the makespan and the total completion time (see \cite{BGJ77}).
Finally, what may not be obvious is  that even if $P||C_{max}$ is strongly \NP-hard, 
the problem $P|w(G)\leq k|C_{max}$ is solvable in pseudo-polynomial time (\cite{MT99}). 
This results is due to the fact that the class of empty graphs is not included in bounded-width graphs class.

\subsection{Open problems}

For an arbitrary number of machines and the makespan criterion, the boundary between polynomially solvable and \NP-hard problems seems very sharp,
we believe that the efforts should not be focused on these problems. When the number of machines is fixed, this boundary is much larger.
Surprisingly, to the best of our knowledge, no other structures of precedence graphs than the ones introduced in previous section 
have been studied for problem $P_m|prec,p_j=p|\Cmax$. In our opinion, it could be a good opportunity 
to work on more general precedence graphs on this problem, to be able to arbitrate if $P_m|prec,p_j=p|\Cmax$ is solvable in polynomial time or \NP-hard. The most natural extension in our opinion
is to consider series-parallel graphs, since it is a generalization of opposing forests, level order graphs and DC-graphs, for which the problem is polynomially solvable.

For the total completion time criterion, the two most intriguing problems are
$P|intree,p_j=p|\sum C_j$ and $P|outtree, r_j,p_j=p|\sum C_j$.
For the former problem, the interest lies in the fact that there exists an ideal schedule for outtree precedences, but not for intree precedences. Nevertheless, we do
believe that this problem admits an optimal polynomial algorithm.
For the latter problem, an algorithm exists for $p_j=1$ (i.e., release dates are multiple of the processing time, see \cite{BHK02}), and hence the
gap to $p_j=p$ seems small.

For a fixed number of uniform parallel machines and the makespan criterion, the set of open problems is wide, since the only polynomial algorithm is for $Q2|chains,p_j=p|\Cmax$ (\cite{BHK99}), and the problem $Qm|prec,r_j,p_j=p|\Lmax$ is still open. 
The same behavior occurs for the total completion time: \cite{DLLV90}  proved that $Qm|r_j,p_j=p|\sum C_j$ is solvable in polynomial time, and problem $Qm|prec,r_j, p_j=p|\sum C_j$ is still open.
We hence believe that this set of problems deserves a deeper study. A first approach may consist in trying
to adapt the algorithms available for identical parallel machines.

\section{Parallel machines with preemption}\label{sec:paraYes}

Timkovsky shows very strong links between preemption and chains, including the fact that
a large set of scheduling problems with preemption can be reduced to problems without preemption,
with UET tasks, and where each job is replaced by a chain of jobs (see Theorem 3.5 in~\cite{T03}).
This interesting result can be applied in many cases, but does not hold for the total completion time criterion.
Moreover, according to the structure of the precedence graph, the resulting graph may not have the same structure.
For example, an intree where each job is replaced by a chain of jobs remains an intree, whereas it does not hold for
interval order graphs (two independent jobs will be replaced by two chains of parallel jobs, which is not an interval order graph).
That is why we will examine more precisely what happens in this section.

\subsection{Polynomial cases}

\subsubsection*{Makespan criterion}

Since $P|tree,p_j=p|\Cmax$ is polynomially solvable (see~\cite{H61}), by Timkovsky's result, so is $P|pmtn,tree|\Cmax$.
The first polynomial algorithm for this problem is proposed in~\cite{MC70} with a running time of $O(n^2)$.
An algorithm in $O(n\log m)$ was then proposed in~\cite{GJ80}. Note that the latter algorithm produces at most $O(n)$
preemptions whereas the former may obtain a schedule with $O(nm)$ preemptions.
\cite{L82} studies the case with release dates and outtree, and shows that it can be solved in $O(n^2)$
with a dynamic priority list algorithm (i.e., priorities may change according to what has already been scheduled).

Timkovsky's result can not be applied to precedence graphs such as interval order graphs.
Yet, it was proven that the problem is also solvable in polynomial time for this precedence structure;
first, \cite{SS89} show it for a fixed number of machines $Pm|pmtn,io|\Cmax$ by using a linear programming approach
based on the set of jobs scheduled at each instant. Later \cite{D99} proposes another linear program that
solves the problem for an arbitrary number of machines, i.e., $P|pmtn,io|\Cmax$.
For a fixed number $m$ of machines, \cite{MQ05} extend the result of \cite{SS89}, by proposing an linear programming
approach for ${\cal A}_m-$order graphs (which properly contain interval order graphs).

\subsubsection*{Other criteria and/or machine environment}

\cite{DLY91} explain that $P2|pmtn,chains|\sum C_j$ is strongly \NP-hard by showing that preemption is useless for this problem, that is why
we will only focus on the UET case for the total completion time criterion. \cite{CM05} prove that $P2|pmtn,prec, p_j=p|\sum C_j$ is solvable in polynomial time by adapting the Coffman-Graham algorithm. Their algorithm finds an ideal schedule. 
By proving that preemption is redundant,
\cite{BT01} prove that $P2|pmtn, outtree, r_j, p_j=1|\sum C_j$ is solvable in polynomial time.
When the precedence graph is an intree, \cite{CDK12} prove that the problem is not ideal, and
\cite{CCDK15} provide a deep analysis of the structure of preemption.
Using the same methodology than \cite{BT01}, \cite{BHK02} show that the problem is solvable in polynomial time with an outtree and an arbitrary number of machines.
Moreover, they provide a $O(n^2)$ algorithm, that admits a $O(n \log n)$ implementation according to~\cite{HL05}.
\cite{L06} slightly improves the result of \cite{BT01} and proposes an algorithm of complexity $O(n^2)$ for the problem
$P2|pmtn,outtree,r_j,p_j=p|\sum C_j$.

\subsection{Minimum {\cal NP}-hard cases}

For the makespan criterion, \cite{U76} shows that $P|pmtn,prec,p_j=p|\Cmax$ is \NP-hard.
For the total completion time criterion, if we carefully look at the \NP-hardness proof of 
$P|prec,p_j=p|\sum C_j$ in~\cite{LRK78},
we can see that preemption is useless for the instance constructed in the reduction  and hence
the preemptive version is still \NP-hard for precedence graphs of bounded height:

\begin{Theorem}
$P|pmtn,prec,p_j=p|\sum C_j$ is \NP-hard even if the precedence graph is of bounded height.
\end{Theorem}

\begin{proof}
We just need to modify slightly the proof in the reduction from {\sc Clique} of \cite{LRK78}.

{\sc Clique} : $G=(V,E)$ is an undirected graph and $k$ an integer. Does $G$ have a clique on $k$ vertices?

Let us recall this reduction; we denote by $v$ (resp. $e$) the number of vertices (resp. edges) of $G$.
We also define following parameters: $l=\frac{k(k-1)}{2}$, $k'=v-k$ and $l'=e-l$ (we use the notations of the original article).
We construct an instance of $P|pmtn,prec,p_j=p|\sum C_j$ with $m=\max \{k,l+k',l'\} +1$ machines and $n=3m$ jobs:
\begin{itemize}
\item for each vertex $i\in V$, there is a job $J_i$.
\item for each edge $[i,j]\in E$, there is a job $J_{[i,j]}$.
\item dummy jobs $J_{h,t}$ with $h\in D_t, t\in\{1,2,3\}, D_1=\{1,\dots,m-k\}, D_2=\{1,\dots,m-l-k'\}, D_3=\{1,\dots, m-l'\}$.
\end{itemize}
There are precedence constraints between any job $J_{g,t}$ and $J_{h,t+1}$, for any $g\in D_t, h\in D_{t+1}, t=1,2$.
Moreover, for any edge $[i,j]\in E$, there is a precedence between $J_i$ and $J_{[i,j]}$.
Finally, the question is whether there is a schedule such that $\sum C_j\leq 6m$.

Clearly, if {\sc Clique} has a solution, so is the scheduling problem: 
\begin{itemize}
\item dummy jobs $J_{h,t}$ are scheduled during the time-interval $[t-1,t]$,
\item the $k$ jobs corresponding to the vertices of the clique are scheduled during the first period,
\item the $l$ jobs corresponding to the edges of the clique, and the $k'$ jobs corresponding to the remaining vertices are scheduled at the second period,
\item all the remaining jobs are scheduled at the third period.
\end{itemize}
This solution has a total completion time of exactly $6m$. Note that it does not use preemption, and is such that $\Cmax=3$.

Conversly, let us prove that if there is a schedule such that  $\sum C_j\leq 6m$, then there is a clique of size $k$.
First, we can easily prove that if there is a schedule such that  $\sum C_j\leq 6m$, then there is a schedule such that $\Cmax \leq 3$. Indeed, if there exists a job $j*$ such that $C_{j*}>3$, then in the best case (i.e., if there is no precedence constraint), we have $\sum C_j \geq 1*m+2*m+3*(m-1)+C_{j*} > 6m$.
We hence know that there is no idle time and that the dummy jobs $J_{h,t}$ are scheduled during interval $[t-1,t]$.
To conclude, we just need to see that if there is no clique of size $k$ then, whatever the schedule on interval $[0,1]$ without idle time, 
the number of eligible jobs at time $1$ is strictly less than $k'+l$, which implies an idle time and hence no schedule such that $\Cmax \leq 3$.

\end{proof}

\subsection{Open problems}

Preemptive parallel machine scheduling problems did not receive as much attention as their non-preemptive counterpart,
hence the set of open problems is wider.

For the makespan objective and an arbitrary number of machines, when the precedence graph is an interval order graph, it is known to be solvable in polynomial time,
as in the UET non-preemptive case. Since for the UET non-preemptive problem, new classes strictly including interval order graphs (namely 
quasi-interval order graphs and over-interval order graphs) have led to polynomial algorithms, the same question arises for 
$P|pmtn,qio|\Cmax$ and $P|pmtn,oio|\Cmax$.

When the number of machines is fixed, there is a wide set of open problems, since $Qm|pmtn,prec,r_j|\Lmax$ is the maximal open problem.
A good challenge may be for example to try to fix the complexity of $Pm|pmtn,prec,p_j=p|\Cmax$, or at least to try to 
find new precedence graphs for which the problem is polynomial, by taking advantage of the fact that tasks are UET (even if it may
not always be helpful with preemption).


For the total completion time criterion, $P|pmtn,outtree,p_j=p|\sum C_j$ is maximal polynomially solvable, and 
we just show that $P|pmtn,prec,p_j=p|\sum C_j$ is \NP-hard even if the precedence graph is of bounded height.
It could be interesting to consider other structures of precedence graph to derive polynomial algorithms.
In a similar way, \cite{BBKT04} show that $P|pmtn, p_j=p|\sum T_j$ is solvable in polynomial time, and adding
precedence constraints makes the problem \NP-hard, but there is no other result available in the literature, it hence
would be interesting to search for new polynomial cases by testing different precedence graphs.

For the problem $P2|pmtn, prec,r_j,p_j=1|\sum C_j$, the gap is even wider: it is polynomial if the precedence graph is an outtree (see~\cite{BT01})
but all the other cases are open, from $P2|pmtn, intree, r_j, p_j = 1|\sum C_j$ to $P2|pmtn, prec, r_j, p_j = 1|\sum C_j$.

\section{Conclusion} \label{sec:ccl}

In this paper, we surveyed the complexity results for scheduling problems with precedence constraints,
and we have seen that single machine scheduling problems have been
much more studied than others. This looks quite normal since the single machine problem is the
scheduling problem that is the closest to order theory. Nevertheless, we show that
there still are a few open problems for the single machine case.
We believe that the most interesting problems for which the complexity is open lie
in the parallel machine case; more precisely, we do conjecture that $Pm|sp-graph,p_j=p|\Cmax$
is solvable in polynomial time; this result will be a large breakthrough since series-parallel graphs
are most of the time studied for single machine problems.

Another approach to understand the complexity of scheduling problems is to deal with
the parameterized complexity (see \cite{DF12}).
There are only very few results on parameterized complexity of scheduling problems. 
One can cite \cite{FM03} who show that, if the precedence graph is of bounded width $w$ (it is equal 
to the size of a maximum antichain), then problem $1|prec|\sum T_j \leq k$ is FPT when
parameterized by $(w,k)$.
The most recent result on the subject is that $P||\Cmax$ is FPT for parameter $p_{max}=\max p_j$ (see \cite{MW14}).
A good graph measure is a powerful tool for parameterized complexity.
For general (undirected) graphs, the creation of the treewidth (see \cite{RS86})
helped to discover many results in graph theory, including the Courcelle's theorem (\cite{C90}). 
For directed graphs, and more specifically DAGs, none of the existing measures (see \cite{GHKLOR14}) is satisfactory. 
In our opinion, a major breakthrough will be achieved when one will be able to find a good measure
on directed acyclic graphs.

\bibliographystyle{plainnat}
\bibliography{Biblio}

\newpage
\section*{Appendix A: List of results}

For an easier reading of all the complexity results that are reviewed in this survey, we proposed a synthesis in the following tables.
In each table, we write the polynomial cases, some open cases (the ones that seem the most promising in our opinion) and the \NP-hard problems.

\begin{center}
\begin{table}[h!]
\begin{center}
\begin{tabular}{|l|l|l|}
\hline
 Problem & Complexity & Reference \\
\hline
	$1|sp-graph|\sum w_j C_j $ & \P & \cite{L78} \\
	$1|dim\leq 2|\sum w_j C_j$ & \P & \cite{AM09}\\
	$1|sp-graph,p_j^{(t)}=p_j(1+at)|\sum w_jC_j$ & \P & \cite{WNC08} \\
	$1|sp-graph,p_j{(t)}=p_j+\alpha_jt|\Cmax$ & \P & \cite{WNC08}\\
	$1|sp-graph,p_j^{[r]}=a_j+b_j r|\Cmax$ & \P & \cite{WW13}\\
	$1|sp-graph,p_j^{[r]}=a_j-b_j r|\Cmax$ & \P & \cite{WW13}\\
	$1|sp-graph,p_j^{[r]}=p_jr|\Cmax$ & \P & \cite{GPSW08}\\
	$1|sp-graph,p_j^{[r]}=p_j\gamma^{r-1}|\sum C_j$ with $\gamma \geq 2$ or $0<\gamma<1$ & \P & \cite{GPSW08}\\
	$1|sp-graph,p_j^{(t)}=p_j(1-at)|\sum w_jC_j$ & \P & \cite{GPSW08} \\
	$1|sp-graph,p_j^{(t)}=p_j+at|\sum C_j$ & \P & \cite{GPSW08} \\
	$1|sp-graph,p_j^{(t)}=p_j-at|\sum C_j$ & \P & \cite{GPSW08} \\
	$1|pmtn, r_j, p_j=p|\sum T_j$ & \P & \cite{TNC06} \\	
	$1|pmtn, r_j, p_j=p|\sum w_j U_j$ & \P & \cite{B99} \\	
	&&\\
	
	$1|2 < fdim < 3|\sum w_jC_j$ & Open &\\
	$1|io|\sum C_j$ & Open & \\
	$1|pmtn, chains, r_j, p_j=p|\sum T_j$ & Open &  \\	
	$1|pmtn, prec, r_j, p_j=p|\sum w_j T_j$ & Open &  \\	
	$1|pmtn, chains, r_j, p_j=p|\sum U_j$ & Open &  \\	
	$1|pmtn, prec, r_j, p_j=p|\sum w_j U_j$ & Open &  \\	
	
	&&\\

	$1|prec,p_j=p|\sum w_j C_j $ & \NP-hard & \cite{L78} \\
	$1|prec|\sum C_j $ & \NP-hard & \cite{L78} \\
	$1|indegree\leq 2|\sum C_j$ & \NP-hard & \cite{LRK78}\\
	$1|h(G)\leq k|\sum w_jC_j$ & \NP-hard & \cite{L78}\\
	$1|io|\sum w_jC_j$ & \NP-hard & \cite{AMMS11}\\
	$1|fdim\geq 3|\sum w_jC_j$ & \NP-hard & \cite{AMMS11}\\
	$1|chains,p_j=1|\sum U_j$ & \NP-hard & \cite{LRK80}\\
	$1|chains,p_j=1|\sum T_j$ & \NP-hard & \cite{LY90}\\
	\hline
\end{tabular}
\caption{Complexity results for single machine problems}
\label{tab:single}
\end{center}
\end{table}
\end{center}

\begin{center}
\begin{table}[h!]
\begin{center}
\begin{tabular}{|l|l|l|}
\hline
 Problem & Complexity & Reference \\
\hline
	$P|tree,p_j=p|\Cmax$& \P & \cite{H61} \\
	$P|outtree,p_j=p|\sum C_j$& \P & \cite{H61} \\
	$Pm|opp. forest,p_j=p|\Cmax$& \P & \cite{GJTY83} \\
	$Pm|lo,p_j=p|\Cmax$& \P & \cite{DW85} \\
	$P|io,p_j=p|\Cmax$& \P & \cite{PY79} \\
	$P|io,p_j=p|\sum C_j$& \P & \cite{M89} \\
	$P|qio,p_j=p|\Cmax$& \P & \cite{M99} \\
	$P|oio,p_j=p|\Cmax$& \P & \cite{CM05} \\
	$P|oio,p_j=p|\sum C_j$& \P & \cite{W15} \\
	$P|outtree, r_j,p_j=p|\Cmax$& \P & \cite{BGJ77} \\
	$P|DC-graph,p_j=p|\Cmax$& \P & \cite{KRP09} \\
	$Pm|h(G)\leq k,p_j=p|\Cmax$& \P & \cite{DW84} \\
	$P|w(G) \leq k, r_j, p_j=1|f$ & \P & \cite{MT99}\\
	$P2|prec, p_j=p|\Cmax$ & \P & \cite{CG72}\\
	$P2|prec,p_j=p|\sum C_j$ & \P & \cite{CG72}\\
	$P2|prec,r_j,p_j=1|\sum C_j$ & \P & \cite{B04}\\
	$Q2|chains,p_j=p|\Cmax$ & \P & \cite{BHK99}\\
	$Pm|intree,p_j=p|\sum C_j$& \P & \cite{BBKT04} \\
	$P|outtree, r_j,p_j=1|\sum C_j$& \P & \cite{BHK02} \\
	$Qm|r_j, p_j=p|\sum C_j$ & \P & \cite{DLLV90}\\

	&&\\

	$Pm|sp-graph, p_j=p|\Cmax$ & Open & \\
	$Pm|prec, p_j=p|\Cmax$ & Open & \\
	$P|intree,p_j=p|\sum C_j$& Open &  \\
	$P|outtree, r_j,p_j=p|\sum C_j$& Open & \\
	$Qm|prec,r_j,p_j=p|\Lmax$ & Open &\\
	$Qm|prec,r_j,p_j=p|\sum C_j$ & Open & \\
	&&\\
	$P2||\Cmax$ & \NP-hard & \cite{LRKB77}\\
	$P2|chains|\sum C_j$ & \NP-hard & \cite{DLY91}\\
	$P|opp. forest,p_j=p|\Cmax$	 & \NP-hard & \cite{GJTY83} \\
	$P|h(G)\leq k,p_j=p|\Cmax$	 & \NP-hard & \cite{LRK78} \\
	$P|intree, r_j,p_j=p|\Cmax$& \NP-hard & \cite{BGJ77} \\
	$P|h(G)\leq k,p_j=p|\sum C_j$	 & \NP-hard & \cite{LRK78} \\
	$P|intree, r_j,p_j=1|\sum C_j$& \NP-hard &  \cite{L} \\

	\hline
\end{tabular}
\caption{Complexity results for parallel machine problems without preemption}
\label{tab:single}
\end{center}
\end{table}
\end{center}

\begin{center}
\begin{table}[h!]
\begin{center}
\begin{tabular}{|l|l|l|}
\hline
 Problem & Complexity & Reference \\
\hline
	$P|pmtn,tree|\Cmax$ & \P & \cite{MC70} \\
	$P|pmtn,outtree,r_j|\Cmax $ & \P & \cite{L82} \\
	$P|pmtn,io|\Cmax$ & \P & \cite{D99}\\
	$Pm|pmtn,{\cal A}_m|\Cmax$ & \P & \cite{MQ05}\\
	$P2|pmtn,outtree, r_j,p_j=1|\sum C_j$ & \P & \cite{BT01}\\
	$P2|pmtn,outtree, r_j,p_j=p|\sum C_j$ & \P & \cite{L06}\\
	$P|pmtn,outtree, r_j,p_j=1|\sum C_j$ & \P & \cite{BHK02}\\
	$P2|pmtn,prec,p_j=1|\sum C_j$ & \P & \cite{C03}\\
	&&\\
	
	$Pm|pmtn,qio, p_j=p|\Cmax$ & Open & \\
	$Pm|pmtn,oio, p_j=p|\Cmax$ & Open & \\
	$Pm|pmtn,lo, p_j=p|\Cmax$ & Open & \\
	$Pm|pmtn,sp-graph, p_j=p|\Cmax$ & Open & \\
	$Pm|pmtn,prec, p_j=p|\Cmax$ & Open & \\
	$Qm|pmtn,prec,r_j|\Lmax$ & Open & \\
	$P2|pmtn, intree, r_j, p_j = 1|\sum C_j$ & Open & \\
	$P2|pmtn, prec, r_j, p_j = 1|\sum C_j$ & Open & \\

	&&\\
	$P|pmtn,prec,p_j=p|\Cmax$ & \NP-hard & \cite{U76} \\
	$P2|pmtn,chains|\sum C_j$ & \NP-hard & \cite{DLY91}\\
	$P|pmtn,h(G)\leq k,p_j=p|\sum C_j$ & \NP-hard & [this paper]\\
	$P2|pmtn,chains,p_j=1|\sum w_j C_j$ & \NP-hard & \cite{DLY91}\\
	$P2|pmtn,chains,p_j=1|\sum U_j$ & \NP-hard & \cite{BBKT04}\\
	\hline
\end{tabular}
\caption{Complexity results for parallel machine problems with preemption}
\label{tab:single}
\end{center}
\end{table}
\end{center}

\end{document}